\documentclass[a4paper,11pt]{article}
\usepackage{amsmath}
\usepackage{amsthm}
\usepackage{amssymb}
\usepackage{amsfonts}
\usepackage{amscd}
\usepackage[latin1]{inputenc}

\theoremstyle{theorem}
\newtheorem{theorem}{Theorem}

\newtheorem{lemma}[theorem]{Lemma}

\theoremstyle{definition}
\newtheorem{definition}[theorem]{Definition}

\theoremstyle{remark}
\newtheorem{remark}[theorem]{Remark}

\newcommand{\al}{\alpha}
\newcommand{\be}{\beta}
\newcommand{\de}{\delta}
\newcommand{\ep}{\epsilon}

\newcommand{\ga}{\gamma}

\newcommand{\la}{\lambda}
\newcommand{\om}{\omega}

\newcommand{\vp}{\varphi}

\newcommand\Om\Omega
\newcommand\Te\Theta
\newcommand{\De}{\Delta}
\newcommand{\Ga}{\Gamma}
\newcommand{\La}{\Lambda}
\newcommand{\Si}{\Sigma}
\newcommand{\bS}{\Sigma}

\newcommand{\tmu}{\tilde{\mu}}

\def\CC{\mathbb{C}}
\def\NN{\mathbb{N}}
\def\RR{\mathbb{R}}

\newcommand{\cD}{{\mathcal D}}

\newcommand{\cM}{{\mathcal M}}

\newcommand{\cO}{{\mathcal O}}

\newcommand{\cR}{{\mathcal R}}

\newcommand{\cT}{{\mathcal T}}

\def\Bx{\bar x}

\def\BM{\,\overline{\!M}{}}

\newcommand{\pd}{\partial}
\newcommand\minus\backslash

\newcommand\lan\langle
\newcommand\ran\rangle

\newcommand{\I}{{\mathrm i}}

\newcommand{\dd}{{\mathrm d}}
    \DeclareMathOperator\Real{Re}

\renewcommand\SS{\mathbb S}

\newcommand\hf{\hat f}

\renewcommand\geq\geqslant
\renewcommand\leq\leqslant

\newcommand\hcM{\widehat\cM}
\DeclareMathOperator{\diam}{diam}

\newcommand\sk{\sqrt k}

\newcommand\hT{\widehat T}

\newcommand\tW{\widetilde W}
\def\Bmu{{\bar\mu}}
\def\Bla{{\bar\la}}
\def\BcM{{\bar{\cM}}}
\def\Btheta{{\bar\vartheta}}
\def\Bw{{w_{\BM}}}
\def\Br{{\bar r}}
\def\cMo{\cM^{(0)}}
\def\tcR{\widetilde{\cR}}
\title{\Large\bf Spinor Green's functions via spherical means\\ on products of space forms}
\author{Alberto Enciso$^a$\thanks{alberto.enciso@math.ethz.ch}\;
  and Niky Kamran$^b$\thanks{nkamran@math.mcgill.ca}}
\date{\small $^a$ Departement Mathematik, ETH Zürich, 8092 Zürich, Switzerland.\vspace{1ex}\\
$^b$ Department of Mathematics and Statistics, McGill University, Montréal, Québec,\\ Canada H3A 2K6.}

\begin{document}
\maketitle

\begin{abstract}
  We explicitly compute the Green's function of the spinor Klein--Gordon equation on the Riemannian and Lorentzian manifolds of the form $M_0×\cdots× M_N$, with each factor being a space of constant
  sectional curvature. Our approach is based on an extension of the method of spherical means to the case of spinor fields and on the use of Riesz distributions.
\end{abstract}

\section{Introduction}

Our objective in this paper is to compute the Green's function for
the spinor Klein-Gordon operator on all Riemannian and Lorentzian
product manifolds of the form $M_0×\cdots× M_N$, where each factor
is a simply connected Riemannian or Lorentzian manifold of constant
sectional curvature. For $N=1$, this class of manifolds includes the
de Sitter and anti-de Sitter spaces, which are conformally flat. In contrast, for $N\geq 2$,  
these manifolds are not conformally flat. They comprise for example 
the Robinson--Bertotti solutions of the Einstein--Maxwell equations, which are of Petrov type D.

To the best of our knowledge, the Green's function for the Dirac equation
on de Sitter and anti-de Sitter spaces was first computed at the
formal level by Allen and Lütken~\cite{Al86} and Mück~\cite{JPA}.
These results have been used in a number of recent cosmological and
field-theoretic applications,
e.g.~\cite{AL03,Co05,BU06,FSSY06,Gr09,Ko09}. Camporesi has also
studied the spinor heat kernel in these spaces~\cite{Ca92}, and
Inagaki, Ishikawa and Muta~\cite{IIM96} have considered the Dirac
equation on the Einstein universe $\RR^1_1×\SS^{n-1}$.

The case $N\geq 2$ treated in this paper differs in significant ways
from the conformally flat case $N=1$ and requires notably a more
elaborate treatment of the underlying spectral theory. Our approach
is based on an extension to the case of spinor fields of the
classical method of spherical means, which is a fundamental tool in
the analysis of scalar wave equations~\cite{He84}. In the Riemannian
case, this leads to an integral representation for the unique
square-integrable solution of the inhomogeneous spinor Klein-Gordon
equation in the class of manifolds under consideration. In the
Lorentzian case, the method of spherical means adapted to spinor
fields leads naturally to the use of Riesz transforms, which are
very well suited for the analysis of linear wave equations. In
addition, the use of spherical means avoids the necessity of
starting with an explicit ansatz for the action of the Dirac
operator on the Green's function, in contrast to the earlier works
cited above for the case of de Sitter and anti-de Sitter space. The
outcome of the calculation is an exact integral representation for
the unique advanced Green's function for the spinor Klein-Gordon
operator in the manifolds considered in this paper.

We should point out that we have recently used the  
method of spherical means to compute the Green's function for the Hodge
Laplacian on differential forms on a class of symmetric spaces which
is a special subclass of the one considered in the present
paper~\cite{EK09}, namely products of Riemannian surfaces of
constant curvature with de Sitter or anti-de Sitter space. For the
case of Hodge Laplacians, the method of spherical means does not
extend beyond the case of two factors in the product structure, with
one factor being a Riemannian surface. We shall see in the present
paper that in contrast with the Hodge Laplacian case, there is no
such restriction for the spinor Klein--Gordon equation and that the
method of spherical means applies to the general product structure.

We have chosen to analyze the spinor Klein--Gordon (or Laplace)
equation~\cite{Ma67} rather than the Dirac equation because the
former presents a remarkable intertwining property with spherical
means that allows us to compute the Green's function not only for
spaces of constant curvature, but also for arbitrary products
thereof. This property, which is presented in Lemmas~\ref{L.DT}
and~\ref{L.DT2}, is made manifest by computing the action of the
spinor Laplacian on spherical means. It is not shared by the Dirac
operator or the Laplacian on forms, which accounts for the fact
mentioned earlier that spherical means can be used to compute the
differential form Green's function on product spaces only when there
is at most one factor of dimension higher
than two. It should be mentioned, however, that the approach
developed in this paper can be readily used to rigorously compute
the Green's function of the Dirac equation on space forms as well.

The paper is organized as follows. In Section~\ref{S.means}, we
introduce two spherical means operators on spaces of constant curvature
and derive a manageable formula for the action of the Dirac operator on
spherically averaged spinor fields on these spaces in terms of a ``radial''
ordinary differential operator. The spectral resolution of this radial operator
plays a key role in our approach and is computed in Section~\ref{S.radial}.
In Section~\ref{S.Riem} we consider the Riemannian case, computing the minimal
spinor Green's function in closed form. The Lorentzian case is dealt with in Section~\ref{S.Lor},
where we explicitly compute the advanced (or retarded) Green's function.

\section{Spherical means for spinors}
\label{S.means}

\subsection{Riemannian spherical means}
\label{SS.Riem}

Throughout this subsection, we shall denote by
$M$ the simply connected Riemannian manifold of dimension $n$ and
constant sectional curvature $k$. It is well known that this manifold admits a unique spin structure. The space of $C^\infty$ spinor fields on $M$ will be denoted by $\bS(M)$, with the fiber of the spinor bundle at a point $x\in M$ being $\bS_x(M)$. Clifford multiplication will be simply denoted by juxtaposition, which allows us to write the action of the Dirac operator in local coordinates $x^i$ as
\[
D_M\psi:=g^{ij}e_i\nabla_j\psi\,.
\]
Here $g$ is the metric in $M$, $\nabla$ is the covariant derivative operator and $e_i$ is a convenient shorthand notation for the vector field $\pd/\pd x^i$.

Let us denote by $\rho_M(x,x')$ the Riemannian distance between two
points $x,x'\in M$. We recall that the injectivity radius of $M$
coincides with its diameter $\diam M$, which is $+\infty$ if $k\leq0$ and
$\pi/\sqrt{k}$ if $k>0$. If $\rho_M(x,x')$ is smaller that the
injectivity radius of $M$, there exists a unique minimal geodesic
$\ga_{x,x'}:[0,1]\to M$ with $\ga_{x,x'}(0)=x'$ and
$\ga_{x,x'}(1)=x$. In this case, we denote by $\nu(x,x')\in T_xM$ the
unit tangent vector to this curve at $x$, which can also be defined as
$\nu(x,x'):=\nabla \rho_M(x,x')$, i.e., the gradient of
$\rho_M(x,x')$ computed with respect to its first variable. The spinor
parallel transport operator from $x'$ to $x$ is the map
$\La(x,x'):\bS_{x'}(M)\to \bS_x(M)$ defined by
$\La(x,x')\psi:=\Psi(1)$, where $\Psi(t)$ is the unique
solution of the parallel transport equation along $\ga_{x,x'}$, which
can be written (with a slight abuse of notation) as
\[
\nabla_{\dot\ga_{x,x'}(t)}\Psi(t)=0\,,\qquad \Psi(0)=\psi\,.
\]
$\La(x,x')$ is clearly well defined, and smoothly dependent on $x$ and $x'$, whenever $x'$ does not belong to the cut locus of $x$.

Our goal for this section is to introduce appropriate spherical means operators on spinors and prove an intertwining relation which allows to exchange the action of the square Dirac operator $D_M^2$ for that of ordinary differential operators. To do so, we start by introducing the notation $B_M(x,r)$ for the ball of center $x$ and radius $r$ and setting $S_M(x,r):=\pd B_M(x,r)$.
It is well known that
$S_M(x,r)$ is diffeomorphic to a sphere for any $0<r<\diam M$ and
that these spheres foliate $M\minus\{x\}$ (and minus the antipodal point
of $x$ if $M$ is positively curved). As the area of $S_M(x,r)$ is
\begin{equation}\label{area}
m_M(r):=\big|\SS^{n-1}\big|\Bigg(\frac{\sin \sk\, r}{\sk}\Bigg)^{n-1}\,,
\end{equation}
we will denote by $\dd\mu_M(r):=m_M(r)\,\dd r$ the radial measure on $M$. The determination of the square root we shall use has non-negative real part, with a branch cut on the negative real axis, and is holomorphic in
$\CC\minus(-\infty,0]$. If $k$ is negative, we take the square root $\sqrt k\in\I\RR^+$.

\begin{definition}\label{D.Riem}
Let $0<r<\diam M$. The (Riemannian) {\em spherical means} of a
smooth spinor field  $\psi \in \bS(M)$ on the sphere of radius $r$ are
defined as the integrals
\begin{align*}
\cM_r\psi(x)&:=\frac{1}{m_M(r)}\int_{S_M(x,r)}\La(x,x')\psi(x')\,\dd
S(x')\,,\\
\hcM_r\psi(x)&:=\frac{1}{m_M(r)}\int_{S_M(x,r)}\nu(x,x')\,\La(x,x')\psi(x')\,\dd
S(x')\,,
\end{align*}
with $\dd S$ standing for the induced hypersurface measure.
\end{definition}

The crucial intertwining property of the spinor spherical means operators in $M$ with the spinor Laplacian is presented in Lemma~\ref{L.DT} below. In order to state the result we introduce the differential operators
\begin{subequations}\label{TtT}
  \begin{align}\label{T}
    T_Mf(r)&:=f''(r)+(n-1)\sk\,\cot\big(\sk r\big)\,f'(r)-\frac{k(n-1)}4\bigg( \sec^2\frac{\sk r}2+n-1\bigg)f(r)\,,\\
 \hT_Mf(r)&:=f''(r)+(n-1)\sk\,\cot\big(\sk r\big)\,f'(r)-\frac{k(n-1)}4\bigg( \csc^2\frac{\sk r}2+n-1\bigg)f(r)\,,
\end{align}
\end{subequations}
and recall the well known formulas for the covariant derivatives of $\nu(x,x')$ and $\La(x,x')$ (cf.\ e.g.~\cite{Al86,PRD,JPA})
\begin{equation}\label{ders}
  \nabla_i \nu_j=\sk\, \cot\big(\sk\rho_M\big)\,(g_{ij}-\nu_i\nu_j)\,,\qquad \nabla_i\La=-\frac{\sk}2\tan\frac{\sk \rho_M}2\,(e_i\nu-\nu_i)\La\,.
\end{equation}
For the ease of notation, we omit the argument $(x,x')$ in $\nu$, $\rho_M$ and $\La$.

\begin{lemma}\label{L.DT}
  For any $f\in C^\infty_0((0,\diam M))$ and $\psi\in\bS(M)$,
  \[
  D_M^2\int_0^\infty \big(f(r)\,\cM_r\psi+\hf(r)\,\cM_r\psi\big)\,\dd\mu_M(r)= \int_0^\infty \big(T_Mf(r)\,\cM_r\psi+\hT_M\hf(r)\,\cM_r\psi\big)\,\dd\mu_M(r)
  \]
\end{lemma}
\begin{proof}
  Let us set
  \begin{align*}
    \Psi(x)&:=\int_0^\infty\big(f(r)\,\cM_r\psi(x)+\hf(r)\,\cM_r\psi(x)\big)\,\dd\mu_M(r)\\
    &\phantom{:}=\int\big(f(\rho_M(x,x'))\,\La(x,x')\psi(x')+\hf(\rho_M(x,x'))\,\nu(x,x')\,\La(x,x')\psi(x')\big)\,\dd x'\,,
  \end{align*}
  where $\dd x'$ stands for the Riemannian volume measure on $M$. Omitting the dependence on $(x,x')$ in the rest of the proof, a short computation using~\eqref{ders} and the facts that $e_ie_j+e_je_i=2g_{ij}$ and $\nu\nu=1$ gives
  \begin{subequations}\label{ders2}
  \begin{align}
    D_M\nu&=g^{kl}e_k\nabla_l\nu_je_ig^{ij}=\sk\, \cot\big(\sk\rho_M\big) \,(g^{ij}e_ie_j -\nu\nu)= (n-1)\,\sk\, \cot\big(\sk\rho_M\big)\,,\\
    g^{ij}e_i\nu\nabla_j\La&=-\frac{\sk}2\tan\frac{\sk \rho_M}2\,g^{ij}(2\nu_i-\nu e_i)(e_j\nu-\nu_j)\La= (n-1)\,\frac{\sk}2\tan\frac{\sk \rho_M}2\,.
  \end{align}
\end{subequations}
We now make use of some well known formulas for the covariant derivatives of $\nu(x,x')$ and $\La(x,x')$ (cf.\ e.g.~\cite{Al86,PRD,JPA})
\begin{equation}\label{ders}
  \nabla_i \nu_j=\sk\, \cot\big(\sk\rho_M\big)\,(g_{ij}-\nu_i\nu_j)\,,\qquad \nabla_i\La=-\frac{\sk}2\tan\frac{\sk \rho_M}2\,(e_i\nu-\nu_i)\La\,.
\end{equation}
For the ease of notation, we omit the argument $(x,x')$ in $\nu$, $\rho_M$ and $\La$.

Using the formulas~\eqref{ders2} and~\eqref{ders}, the action of the Dirac operator on $\Psi$ is readily found to be
\begin{align*}
  D_M\Psi&=\int\big(D_M(f(\rho_M)\La\psi)+D_M(\hf\nu\La\psi)\big)\,\dd x'\\
  &=\int\big(f'(\rho_M)\nu\La\psi+f(\rho_M)D_M\La\psi+\hf'(\rho_M)\La\psi+\hf(\rho_M)(D_M\nu)\La\psi +\hf(\rho_M) g^{ij}e_i\nu\nabla_j\La\psi\big)\,\dd x'\\
  &=\int\big(f'(\rho_M)-(n-1)\,\frac{\sk}2\tan\frac{\sk \rho_M}2\,f(\rho_M)\big)\,\nu\La\psi\,\dd x' \\
  &\qquad\qquad\qquad+ \int\bigg(\hf'(\rho_M)+(n-1)\,\sk\,\bigg( \cot\big(\sk\rho_M\big)\,+\frac12\tan\frac{\sk \rho_M}2\bigg)\,f(\rho_M)\bigg)\,\nu\La\psi\,\dd x'\,.
\end{align*}
Acting with $D_M$ on both sides of this equation and arguing as above we readily find
\[
  D_M^2\Psi=\int\big( T_Mf(\rho_M)\La\psi+\hT_M\hf(\rho_M)\nu\La\psi\big)\,\dd x'=\int \big(T_Mf(r)\,\cM_r\psi+\hT_M\hf(r)\,\cM_r\psi\big)\,\dd\mu_M(r)\,,
\]
as claimed.
\end{proof}

\subsection{Lorentzian spherical means}
\label{SS.Lor}

We now show how the above results must be modified when the underlying manifold is Lorentzian. To begin with, let $M$ now be the simply connected Lorentzian manifold of dimension $n$ and constant sectional curvature $k$. The Minkowski ($k=0$) and de Sitter spaces ($k>0$)
are globally hyperbolic. Anti-de Sitter space ($k<0$) is not, although it is strongly causal. We shall therefore restrict ourselves to a domain $\Om\subset M$ which will be chosen be geodesically normal  (that is a normal neighborhood of each of its points) when $k<0$, while when $k\geq0$, $\Om$ can be taken to be the whole space $M$.

We use the standard notation $J^+_\Om(x)$ for the set of points in $\Om$
causally connected with $x$, that is  the set of points $x'\in\Om$ such that
there exists a future directed, non-spacelike curve from $x$ to
$x'$. Likewise, we set
\begin{align*}
S^+_\Om(x,r):=\big\{x'\in J_\Om^+(x):\rho_M(x,x')=r\big\}
\end{align*}
and denote by $\dd S$ the induced hypersurface measure. Here $\rho_M(x,x')$ stands for the distance between two points $x$ and $x'\in J_\Om^+(x)$. It is well known that $J_\Om^+(x)$ is foliated by the level sets of $\rho_M$,
\begin{equation*}
J^+_\Om(x)=\bigcup_{r\geq0}S_\Om^+(x,r)\,,\qquad \text{with }\; S^+_\Om(x,r)\cap S^+_\Om(x,r')=\emptyset\;\text{ if }\;r\neq r'\,.
\end{equation*}
In the Lorentzian case we still define $m_M$ as in~\eqref{area} and set $\dd\mu_M(r):=m_M(r)\,\dd r$. Notice that $m_M(r)$ is not the area of $S^+_\Om(r)$, which can be infinite. If $x\in\Om$ and $x'\in J^+_\Om(x)$
does not lie on the causal cut locus of $x$, the parallel transport
operator $\La(x,x'):\bS_{x'}(\Om)\to \bS_x(\Om)$ along the unique
future-directed minimal geodesic from $x$ to $x'$ is again well defined. We also keep the notation $\nu(x,x')\in T_x\Om$ for the unit tangent vector at $x$ of the aforementioned minimal geodesic.

\begin{definition}\label{D:LSM}
Let $\psi\in \bS(\Om)$ be compactly supported. Its (Lorentzian) {\em spherical means} of radius $r$ are defined as
\begin{align*}
\cM_r^\Om\om(x)&:=\frac{1}{m_M(r)}\int_{S^+_\Om(x,r)}\La(x,x')\psi(x')\,\dd
S(x')\,,\\
\hcM_r^\Om\om(x)&:=\frac{1}{m_M(r)}\int_{S^+_\Om(x,r)}\nu(x,x')\,\La(x,x')\psi(x')\,\dd
S(x')\,.
\end{align*}
\end{definition}

In what follows, we shall omit the superscript $\Om$ for simplicity of notation. The main property of the spherical means is the following analog of Lemma~\ref{L.DT}. We omit its proof, for it follows from Eqs.~\eqref{ders} and~\eqref{ders2} (which, being of an essentially algebraic nature, remain valid in the context of Lorentzian manifolds) just as in the Riemannian case.

\begin{lemma}\label{L.DT2}
   For any $f\in C^\infty_0((0,\diam M))$ and any compactly supported $\psi\in\bS(\Om)$,
  \[
  D_M^2\int_0^\infty \big(f(r)\,\cM_r\psi+\hf(r)\,\cM_r\psi\big)\,\dd\mu_M(r)= \int_0^\infty \big(T_Mf(r)\,\cM_r\psi+\hT_M\hf(r)\,\cM_r\psi\big)\,\dd\mu_M(r)
  \]
\end{lemma}

\section{The radial operator}
\label{S.radial}

In this section we shall compute the spectral resolution of the radial operator~\eqref{T}, which will play a key role in the rest of the paper. For concreteness, we shall assume that $k$ is nonzero, the flat case being obviously easier to handle.

As the differential operator $T_M$ is nonnegative and symmetric in the domain $C^\infty_0((0,\diam M))$, we will denote by $\cT_M$ its Friedrichs extension. The spectral projection of $\cT_M$ is most easily presented in terms of the Borel measure $\dd\vartheta_M$ on $\RR^+$ defined as
\[
\dd\vartheta_M(\la):=\begin{cases}
  \displaystyle\bigg(\sum\limits_{j=0}^\infty \frac{k^{n/2}(2j+n)\,j!\,(j+n-1)!}{2^{n+1}|\SS^{n-1}|\,\Ga(j+\frac n2)\,\Ga(j+\frac n2+1)}\de\big(\la-k(\tfrac n2+j)^2\big)\bigg)\,\dd\la &\text{if }\; k>0\,,\\[5mm]
 \left|\dfrac{\Ga\big(\frac n2+\sqrt{\frac\la k}\big)\,\Ga\big( \sqrt{\frac\la k}\big)}{\Ga(\frac n2)\,\Ga\big(2\sqrt{\frac\la k}\big)}\right|^2
  \dfrac{(-k)^{(n-1)/2}\,\dd\la}{ 2^n|\SS^{n-1}|\sqrt\la }  &\text{if }\; k<0\,.
\end{cases}
\]
In addition to this, let us denote by $w_M:(0,\diam M)× \RR^+\to\RR$ the function defined $(\mu_M×\vartheta_M)$-almost everywhere by requiring that
\[
w_M\bigg(r,\Big(\frac n2+j\Big)^2\bigg):=\cos\frac{\sk r}2\, P^{(\frac n2-1,\frac n2)}_j\big(\cos\sk r\big)
\]
for all $j\in \NN$ if $k>0$ and
\[
w_M(r,\la):=\cos \frac{\sk r}2\, F\bigg(\frac n2+\sqrt{\frac\la k},\frac n2-\sqrt{\frac\la k},\frac n2;\sin^{2}\frac{\sk r}2\bigg)
\]
if $k<0$. Here we are using the standard notation $P^{(\al,\be)}_j(z)$ and $F(\al,\be,\ga;z)$ for the Jacobi polynomials and the Gauss hypergeometric function.

All the information we shall need concerning the spectral properties of $\cT_M$ is summarized in the following

\begin{theorem}\label{T.TM}
With $\vartheta_M$ and $w_M$ defined as above, the following statements hold:
\begin{enumerate}
\item $w_M(\cdot,\la)$ is an analytic formal eigenfunction of $T_M$ with eigenvalue $-\la$  for $\vartheta_M$-almost every $\la\in \RR^+$
\item The map
\[
U_Mf(\la):=\int_0^{\diam M} f(r)\,{w_M(r,\la)}\,\dd\mu_M(r)
\]
defines a unitary transformation $L^2((0,\diam M),\dd\mu_M)\to L^2(\RR^+,\dd\vartheta_M)$ with inverse given by
\[
U_M^{-1}h(r):=\int_0^{\infty} w_M(r,\la)\,h(\la)\,\dd\vartheta_M(\la)\,..
\]

\item If $g$ is a measurable function and $f\in L^2((0,\diam M),\dd\mu_M)$ is such that $g\, U_Mf\in L^2(\RR^+,\dd\vartheta_M)$,  then
\begin{equation*}
g(-\cT_M)f(r)=\int_0^{\infty} g(\la)\,w_M(r,\la)\,U_Mf(\la)\,\dd\vartheta_M(\la)
\end{equation*}
in the sense of norm convergence.

\item If $f\in L^2((0,\diam),\dd\mu_M)\cap C^\infty([0,\diam M))$, then
\begin{equation}\label{unifconv}
f(r)=\int_0^{\infty} w_M(r,\la)\,U_Mf(\la)\,\dd\vartheta_M(\la)
\end{equation}
for all $r\in [0,\diam M)$.
\end{enumerate}
\end{theorem}
\begin{proof}
We find it convenient to perform the analysis of the operator $\cT_M$ in terms of the variable $z:=\sin^2\frac{\sk r}2$, which ranges over $(0,1)$ when $k$ is positive and over $(-\infty,0)$ when $k$ is negative. If we denote this interval by $I_k$, the latter change of variables induces a unitary isomorphism $L^2((0,\diam M), \dd\mu_M)\to L^2(I_k,\dd\tmu_M)$, with
  \[
\dd\tmu_M(z):=C_M\big[z(1-z)\big]^{\frac n2-1}\,\dd z\,,\qquad C_M:=2^{n-1}k^{-n/2}|\SS^{n-1}|\,.
  \]
  In terms of this new variable, the ordinary differential equation $T_Mw=-\la w$ reads
  \begin{equation}\label{Tf}
z(1-z)\,w''(z)+ \Big(\frac n2-nz\Big)\,w'(z)-\bigg[\frac{n-1}4\Big(n-1+\frac1{1-z}\Big)-\frac\la k\bigg]\,w(z)=0\,.
\end{equation}
This is a Fuchsian differential equation is three regular singular points: $0$ (with characteristic exponents $0$ and $1-\frac n2$), $1$ (with exponents $\frac12$ and $\frac{1-n}2$) and $\infty$ (with exponents $\frac{1-n}2±\sqrt{\la/k}$).

From the knowledge of the characteristic exponents we infer that $T_M$ is essentially self-adjoint when $n\geq4$, the domain of its unique self-adjoint extension $\cT_M$ being
\begin{equation*}
\cD_M:=\big\{w\in C^1(I_k): w' \text{ is absolutely continuous and } T_Mw \in L^2(I_k,\dd\tmu_M)\big\}\,.
\end{equation*}
For $n=2,3$, $T_M$ is of limit-circle type at $z=0$, and its Friedrichs extension is defined on the domain~\cite{MZ00}
\begin{equation*}
\cD_M:=\Big\{w\in C^1(I_k): w' \text{ is absolutely continuous, }\lim_{z\to0}z^{n/2}w(z)=0, \text{  and } T_Mw \in L^2(I_k,\dd\tmu_M)\Big\}\,.
\end{equation*}

Let us begin with the case $k>0$. It is standard that $-\cT_M$ is a positive operator and has a discrete spectrum. Its eigenfunctions are the solutions $w$ of~\eqref{Tf} in $(0,1)$ which are bounded at both endpoints. This is well known to imply that
\[
w(z)=C\, (1-z)^{1/2}\,F\bigg(\frac{n}2+\sqrt{\frac\la k},\frac{n}2-\sqrt{\frac\la k},\frac n2;z\bigg)
\]
is a polynomial multiple of $(1-z)^{1/2}$, which is tantamount to demanding that $\la=(\frac n2+j)^2$ for some nonnegative integer $j$. Therefore,
\[
w(z)=C'\,(1-z)^{1/2}\,P_j^{(\frac n2-1,\frac n2)}(1-2z)
\]
for some normalization constant $C'$. The norm of $w$ is easily computed using that~\cite[22.2.1]{AS70}
\begin{align*}
  \int_0^1|w(z)|^2\,\dd\tmu_M(z)&=C_M|C'|^2\int_0^1z^{\frac n2-1}(1-z)^{n/2}\, P_j^{(\frac n2-1,\frac n2)}(1-2z)^2\,dz\\
  &=4C_M|C'|^2\frac{\Ga(j+\frac n2)\,\Ga(j+\frac n2+1)}{(2j+n)\,j!\,(j+n-1)!}\,.
\end{align*}
In the case of positive curvature, the claim follows immediately from the spectral theorem upon expressing the above formulas in the original variable $r$.

Let us now consider the case $k<0$. From the expression for the exponents at infinity we deduce that, for non-negative $\la$, the eigenvalue equation~\eqref{Tf} does not admit any square integrable solutions at infinity, so the spectrum of $\cT_M$ must be purely (absolutely) continuous. In order to locate the continuous spectrum of the operator, notice that Eq.~\eqref{Tf} can be written as
\[
-\frac1{\om(z)}\frac{\dd}{\dd z}\bigg(\frac{p(z)}{\om(z)}\frac{\dd w}{\dd z}\bigg) +q(z)\,w(z)=0\,,
\]
with
\begin{gather*}
  p(z)=[z(z-1)]^{n-1}\,,\qquad \om(z):=[z(z-1)]^{\frac n2-1}\,,\qquad q(z):=-\frac{n-1}4\bigg(n-1+\frac1{1-z}\bigg)\,.
\end{gather*}
Since the function
\[
Q(z):=q(z)+\frac14\bigg[\frac1{\om(z)}\frac{\dd}{\dd z}\bigg(\frac1{\om(z)}\frac{\dd p}{\dd z}\bigg)- \frac1{4\om(z)^2p(z)}\bigg(\frac{\dd p}{\dd z}\bigg)^2\bigg]= \frac{(n-1) (n+4 z-3)}{16 z(z-1) }
\]
tends to $0$ as $z\to-\infty$, it is a standard result~\cite[Theorem XIII.7.66]{DS88} that the spectrum of $-\cT_M$ is the half-line $[0,\infty)$.

In order to compute the spectral projector of $-\cT_M$, let us take $\la\in\CC\minus[0,\infty)$ with positive real part. Then the solutions of~\eqref{Tf} which are square-integrable in a neighborhood of $0$ and satisfy the boundary condition are proportional to
\[
w_0(z,\la):=(1-z)^{1/2}\, F\bigg(\frac{n}2+\sqrt{\frac\la k},\frac{n}2-\sqrt{\frac\la k},\frac n2;z\bigg)\,,
\]
while the ones that are integrable at infinity at proportional to
\[
w_±(z,\la):=(-z)^{-\frac n2±\sqrt{\frac\la k}}(1-z)^{1/2}\, F\bigg(-\frac n2±\sqrt{\frac\la k},1±\sqrt{\frac\la k},1± 2\sqrt{\frac\la k};\frac1z\bigg)
\]
when the imaginary part of $±\la$ is positive. From the asymptotic behavior these solutions, one can readily infer that the reduced Wronskian of $w_-$ and $w_+$, which is constant, is given by
\[
\tW(w_+(\cdot,\la),w_-(\cdot,\la)):=p(z)\bigg(\frac{\dd w_+}{\dd z}w_-(z)- w_+(z) \frac{\dd w_-}{\dd z}\bigg)=-2\sqrt{\frac\la k}\,.
\]

The above local solutions are connected by the formula~\cite[15.3.7]{AS70}
\begin{equation*}
  w_-(z,\la)=K_0(\la)\,w_0(z,\la)+K_+(\la)\,w_+(z,\la)\,,
\end{equation*}
with
\begin{align*}
  K_0(\la):=\frac{\Ga\big(\frac n2-\sqrt{\frac\la k}\big)\,\Ga\big(-\sqrt{\frac\la k}\big)}{\Ga\big(\frac n2\big)\, \Ga\big(-2\sqrt{\frac\la k}\big)}\,,\qquad K_+(\la):=-\frac{\Ga\big(2\sqrt{\frac\la k}\big)\,\Ga\big(\frac n2-\sqrt{\frac\la k}\big)\, \Ga\big(-\sqrt{\frac\la k}\big)}{\Ga\big(2\sqrt{\frac\la k}\big)\,\Ga\big( \frac n2+\sqrt{\frac\la k}\big)\, \Ga\big(\sqrt{\frac\la k}\big)}\,.
\end{align*}
Since $\sqrt{\la/k}$ is purely imaginary, it is not difficult to see that $K_0(\la)$ and $K_+(\la)$ are homolomorphic functions of $\la$ in the half-plane $\Real\la>0$. Setting
\[
U_Mw(\la):=\int w_0(z,\la)\,w(z)\,\dd\tmu_M(z)
\]
for all $w\in L^2(\RR^-,\dd\tmu_M)$ and realizing that $\dd\tmu_M(z)=|C_M\om(z)|\,\dd z$, we  deduce from the spectral theorem (cf.\ e.g.~\cite[p.\ 1524]{DS88}) that
\[
g(-\cT_M)w(z)=\int_0^\infty g(\la)\, w_0(z,\la)\,U_Mw(\la)\,\dd\vartheta_M(\la)
\]
whenever $g\, U_Mw\in L^2(\RR^+,\dd\vartheta_M)$,
the spectral measure being defined as
\begin{align*}
  \dd\vartheta_M(\la)&:=\frac1{2\pi \I}\frac{K_0(\la)^2}{K_+(\la)\,\tW(w_+(\cdot,\la),w_-(\cdot,\la))}\,\dd\la\,.
\end{align*}
By reverting to the original coordinate $r$ we derive the first three assertions about the spectral decomposition of $-\cT_M$.

Concerning the fourth assertion, we first notice that the validity
of~\eqref{unifconv} for $r\in (0,\diam M)$ is straightforward under
the assumption that $f$ is continuous~\cite[Theorem 6.1]{LS75}. The
fact that it also converges at $0$ is considerably more subtle and
requires higher differentiability assumptions on $f$. To in order to
prove it, it suffices to fix a point $x'\in M$ and notice that $T_M$
is related to the Schrödinger operator
\[
L_M:=\De_M-\frac{k(n-1)}4\bigg( \sec^2\frac{\sk \rho_M(\cdot,x')}2+n-1\bigg)
\]
via
\[
(T_M\vp)(\rho_M(\cdot,x'))=L_M\big(\vp(\rho_M(\cdot,x'))\big)\,.
\]
The convergence of~\eqref{unifconv} at $0$ then follows from the
results of of Pinsky and Taylor~\cite{Pi94,PT97} on the convergence
at $0$ of generalized spherical Fourier transforms.
\end{proof}

\section{The Riemannian case}
\label{S.Riem}

Let $M_a$ ($0\leq a\leq N$) be the simply connected Riemannian manifold of dimension $n_a$ and constant sectional curvature $k_a$, and consider the product space
\begin{equation}\label{BM}
\BM:=M_1×\cdots× M_N\,.
\end{equation}
Obviously the Dirac operator $D_{M_0× \BM}$ in $M_0× \BM$ can be identified with the operator $D_0\oplus D_1\oplus\cdots\oplus D_N$ acting on $\bS(M_0)\oplus \bS(M_1)\oplus\cdots \oplus\bS(M_N)$, where here and in what follows we use the subscript $a$ rather than $M_a$ for the ease of notation. For any ``mass'' parameter $m> 0$, our goal in this section is to solve the spinor equation
\begin{equation}\label{eqR}
  \big(m^2-D_{M_0× \BM}^2\big)\psi=\phi
\end{equation}
explicitly for compactly supported $\phi\in\bS(M_0× \BM)$. Notice that there exists precisely one $L^2$ solution $\psi$ of this equation because the operator $m^2-D_{M_0× \BM}^2$ has a compact inverse. To derive an explicit expression for $\psi$, we will rely on the intertwining relation presented in Lemma~\ref{L.DT} and on the spectral decomposition worked out in Theorem~\ref{T.TM}.

We will denote by $\cM^{(a)}_{r_a}$ the spinor spherical means operator in $M_a$ of radius $r_a$ and introduce the shorthand notation
\begin{subequations}\label{shorthand}
\begin{gather}
  \Br:=(r_1,\dots, r_N)\in \prod_{a=1}^N[0,\diam M_a]\,,\quad \Bla:=(\la_1,\dots,\la_N)\in (\RR^+)^N\,,\quad \Bx:=(x_1,\cdots, x_N)\in\BM\\
  \dd\Bmu(\Br):=\dd\mu_1(r_1)\cdots\dd\mu_N(r_N)\,,\quad \dd\Btheta(\Bla):=\dd\vartheta_1(\la_1)\cdots \dd\vartheta_N(\la_N)\,,\quad  |\Bla|:=\la_1+\cdots+\la_N\,,\\
  \BcM_\Br:=\cM^{(1)}_{r_1}\cdots \cM^{(N)}_{r_N}\,,\quad  \Bw(\Br,\Bla):=w_1(r_1,\la_1)\cdots w_N(r_N,\la_N)\,,\quad U_{\BM}:= U_1\cdots U_N\,.
\end{gather}
\end{subequations}
Before stating and proving the main result of this section, we establish a key lemma. In the statement of this lemma, we shall use the constant $C(n_0,k_0,\la_0)$ defined as
\[
C(n_0,k_0,\la_0):=\frac1{(n_0-2)|\SS^{n_0-1}|}\bigg(\frac{k_0}4\bigg)^{\frac{n_0}2-1} \frac{\Ga\big(\frac{n_0}2 +\sqrt{\frac{\la_0}{k_0}}\big)\, \Ga\big(\frac{n_0}2 -\sqrt{\frac{\la_0}{k_0}}\big)}{\Ga\big(\frac{n_0}2+1\big)\, \widetilde\Ga\big(\frac{n_0}2-1\big)}
\]
if $k_0>0$ and by
\[
C(n_0,k_0,\la_0):=\frac1{2(n_0-2)|\SS^{n_0-1}|}\bigg(\frac{k_0}4\bigg)^{\frac{n_0}2-1} \Bigg[\Real\frac{ \widetilde\Ga\big(\frac{n_0}2-1\big)\, \Ga\big(1+2\sqrt{\frac{\la_0}{k_0}}\big)}{ \Ga\big(\frac{n_0}2 +\sqrt{\frac{\la_0}{k_0}}\big)\, \Ga\big(\frac{1-n_0}2 +\sqrt{\frac{\la_0}{k_0}}\big)}\Bigg]^{-1}
\]
if $k_0<0$. Here
\[
\widetilde\Ga(t):=\begin{cases}
  \Ga(t)\quad &\text{if }\;t\neq0\,,\\
  2&\text{if }\;t=0\,.
\end{cases}
\]

\begin{lemma}\label{L.Green}
  For each $\la_0>0$ such that $\sqrt{\la_0/k_0}-{n_0}/2$ is not a positive integer, define the function
  \[
  g(r_0,\la_0):=C(n_0,k_0,\la_0)\,\cos\frac{\sqrt{k_0}r_0}2\, F\bigg(\frac{n_0}2+\sqrt{\frac{\la_0}{k_0}}, \frac{n_0}2-\sqrt{\frac{\la_0}{k_0}}, \frac{n_0}2+1;\cos^2\frac{\sqrt{k_0}r_0}2\bigg)
  \]
  if $k_0>0$ and
  \begin{multline*}
  \tilde g(r_0,\la_0):=
  2 C(n_0,k_0,\la_0)\,\cos\frac{\sqrt{k_0}r_0}2\\
  \Real\Bigg[ \bigg(\cos\frac{\sqrt{k_0}r_0}2\bigg)^{-n_0-2\sqrt{\frac{\la_0}{k_0}}} F\bigg(\frac{n_0}2+\sqrt{\frac{\la_0}{k_0}}, \sqrt{\frac{\la_0}{k_0}}1+2\sqrt{\frac{\la_0}{k_0}};\cos^{-2}\frac{\sqrt{k_0}r_0}2\bigg)\Bigg]
   \end{multline*}
   if $k_0<0$.  Then $f(\cdot,\la_0)$
  is a real solution of the equation $(\la_0-T_0)g(\cdot,\la_0)=0$ in $(0,\diam M_0]$ and has the asymptotic behavior
  \[
  g(r_0,\la_0)=\begin{cases}\dfrac1{(n_0-2)|\SS^{n_0-1}|\,r_0^{n_0-2}}+ \cO\big(r_0^{1-n_0}\big)\quad &\text{if }\; n_0\geq 3\,,\\[4mm]
    \dfrac1{2\pi}\log\dfrac1{r_0} & \text{if }\; n_0=2\,.
  \end{cases}
  \]
\end{lemma}
\begin{proof} We introduce the new variable $z:=\sin^2\frac{\sqrt{k_0}r_0}2$, in terms of which the equation $(\la_0-T_0)g=0$ reads as
  \begin{equation}\label{eqg}
  z(1-z)\,g''(z)+ \Big(\frac{n_0}2-n_0z\Big)\,g'(z)-\bigg[\frac{n_0-1}4\Big(n_0-1+\frac1{1-z}\Big)-\frac{\la_0}{k_0} \bigg]\,g(z)=0\,.
\end{equation}
For notational simplicity, we omit the dependence of $g$ on $\la_0$. Notice that~\eqref{eqg} is a Fuchsian differential equation with three regular singular points: $0$ (with exponents $0$ and $1-\frac{n_0}2$), $1$ (with exponents $\frac12$ and $\frac{1-n_0}2$) and $\infty$ (with exponents $\frac{n_0}2±\sqrt{{\la_0}/{k_0}}$).

Suppose that $k_0$ is positive, so that $z\in[0,1]$. Since $g$ is regular at $z=1$, there exists a constant $C$ such that
\[
g(z)=C\, (1-z)^{1/2}F(a,b,1+c;1-z)\,,
\]
where we have set
\begin{equation}\label{abc}
  a:=\frac{n_0}2+\sqrt{\frac{\la_0}{k_0}}\,,\quad b:=\frac{n_0}2-\sqrt{\frac{\la_0}{k_0}}\,,\quad c=\frac{n_0}2\,.
\end{equation}
To fix the normalizing constant $C$, it suffices to note~\cite[15.3.6--15.3.12]{AS70} that
\[
F(a,b,1+c;1-z)=\frac{\Ga\big(\frac{n_0}2-1\big)\, \Ga\big(\frac{n_0}2+1\big)}{ \Ga\big(\frac{n_0}2+\sqrt\frac{\la_0}{k_0}\big)\, \Ga\big(\frac{n_0}2-\sqrt\frac{\la_0}{k_0}\big)}\,z^{1-\frac{n_0}2}+o\big(z^{1-\frac{n_0}2}\big)
  \]
  for $z\in(0,1)$ and $n_0\neq2$ and to express the result in the original variable $r_0$. When $n_0=2$, the above equation remains valid if we replace $z^{1-\frac{n_0}2}$ by $\log z$.

  The case $k_0<0$ is similar. A real solution of~\eqref{eqg} is
  \[
g(z)=C\,(1-z)^{1/2}\,\bigg[(1-z)^{-a} F\bigg(a,c-b,a-b+1;\frac1{1-z}\bigg) +(1-z)^{-b} F\bigg(b,c-a,b-a+1;\frac1{1-z}\bigg)\bigg]\,,
\]
with $a,b$ and $c$ defined as in~\eqref{abc}. Using~\cite[15.3.6--15.3.12]{AS70} and arguing as above, one readily arrives at the desired formula.
\end{proof}

\begin{theorem}\label{T.mainR}
  For any $\ep>0$, let $\psi:M_0×\BM$ be the smooth spinor field defined by
  \begin{equation}\label{formpsi}
\psi(x_0,\Bx):=\int f(r_0,\Bla)\,\Bw(\Br,\Bla)\,\cMo_{r_0}\BcM_\Br\phi(x_0,\Bx)\,\dd\mu_0(r_0)\,\dd\Bmu(\Br)\,\dd\Btheta(\Bla)\,,
\end{equation}
where
  \[
f(r_0,\Bla):=\Bw(0,\Bla)\,g(r_0,m^2+|\Bla|)
\]
and $g(r_0,\la_0)$ is defined as in Lemma~\ref{L.Green}. Then $\psi$ is the unique square-integrable solution of $(m^2-D^2_{M_0×\BM})\psi=\phi$.
\end{theorem}
\begin{proof}
  Note that the function $r_0\mapsto \cMo_{r_0}\cM_\Br\phi(x_0,\Bx)$ is in $C^\infty_0([0,\diam M_0])$ for all $(\Br,x_0,\Bx)$ because $\phi$ is smooth and compactly supported. By Lemma~\ref{L.DT} and Theorem~\ref{T.TM}, the action of $m^2-D_{M_0×\BM}^2$ on $\psi$ is given by
  \begin{align*}
   \big(m^2-D_{M_0×\BM}^2\big)\psi&= \int  (m^2-T_0)f(r_0,\Bla)\,\Bw(\Br,\Bla)\,\cMo_{r_0}\BcM_\Br\phi\,\dd\mu_0(r_0)\,\dd\Bmu(\Br)\,\dd\Btheta(\Bla)\\
    & -\int f(r_0,\Bla)\,(T_1+\cdots+ T_N)\Bw(\Br,\Bla) \,\cMo_{r_0}\BcM_\Br\phi\,\dd\mu_0(r_0)\,\dd\Bmu(\Br)\,\dd\Btheta(\Bla)\\
   &=  \int (m^2+|\Bla|-T_0)f(r_0,\Bla)\,\Bw(\Br,\Bla) \,\cMo_{r_0}\BcM_\Br\phi\,\dd\mu_0(r_0)\,\dd\Bmu(\Br)\,\dd\Btheta(\Bla)\,.
 \end{align*}
Lemma~\ref{L.Green} ensures that $f(\cdot,\Bla)$ is a Green's function of the operator $m^2+|\Bla|-T_0$ with pole at~$0$. Using this property, the fact that $\cM^{(a)}_0$ is the identity operator and Theorem~\ref{T.TM}, the above integral reduces to
  \begin{align*}
    \big(m^2-D_{M_0×\BM}^2\big)\psi(x_0,\Bx)&=\int \Bw(0,\Bla)\,\Bw(\Br,\Bla) \,\BcM_\Br\phi(x_0,\Bx)\,\dd\Bmu(\Br)\,\dd\Btheta(\Bla)\\
   &= \int \Bw(0,\Bla)\,U_M\big(\BcM\phi(x_0,\Bx)\big)(\la)\,\dd\Btheta(\Bla)
    =\psi(x_0,\Bx)\,,
  \end{align*}
  as required.
\end{proof}

\begin{remark}
  The case $m=0$ can be tackled with a similar argument under the assumption that $\phi$ is orthogonal to any $L^2$ harmonic spinor in $M_0×\BM$. A theorem of Goette and Semmelmann~\cite{GS01} ensures that nontrivial harmonic spinors exist if and only if $m=0$, the dimensions $n_a$ are even, and the curvatures $k_a$ are negative for all $a\geq0$; further details can be found in~\cite{CP01}. It is also clear that the formula~\eqref{formpsi} remains valid under less stringent integrability and regularity conditions for $\phi$, but we shall not pursue this issue here.
\end{remark}
\begin{remark}
 Define $A:=\{a: k_a>0\}$ and let us write the $a$-th eigenvalue, with $a\in A$, as $\la_a=(\frac{n_a}2+j_a)^2$, where $j_a\in\NN$.  The number $\Bw(0,\Bla)$ appearing in the statement of Theorem~\ref{T.mainR} can be explicitly written as
  \[
\Bw(0,\Bla)=\prod_{a\in A}\frac{\Ga\big(j_a+ \frac{n_a}2\big)}{j_a!\, \Ga\big(\frac{n_a}2\big)}\,.
\]
\end{remark}

\section{The Lorentzian case}
\label{S.Lor}

We now denote by $M_0$ the simply connected Lorentzian space form of dimension $n_0$ and sectional curvature $k_0$, while $M_a$ ($1\leq a\leq N$) stands for the simply connected Riemannian manifold of dimension $n_a$ and constant sectional curvature $k_a$. We will also consider the product manifold $\BM$, defined in~\eqref{BM}, and use the shorthand notation~\eqref{shorthand}.

As discussed in Section~\ref{SS.Lor}, we will restrict ourselves to a normal domain $\Om\subset M_0$, which is a proper subset of $M_0$ if $k_0$ is negative. Given a ``mass'' $m>0$, our goal in this section is to use the machinery developed in the previous sections to solve the spinor wave equation
\begin{equation}\label{eqL}
  \big(m^2-D_{M_0× \BM}^2\big)\psi=\phi
\end{equation}
in closed form, where the spinor field $\phi\in\bS(\Om×\BM)$ is compactly supported. For this purpose we will rely on spinor spherical means, the spectral analysis of the radial operators and the theory of Riesz transforms.

Let us denote by $\cD'(\Om)$ the space of distributions on $\Om$. The (advanced) {\em Riesz distribution} is defined as
\begin{equation}\label{Riesz}
R^\al_{\Om,x_0}(h):=\int_{J^+_{\Om}(x_0)}g_\al\big(\rho(x_0,x_0')\big)^{\al-n_0}\,h(x_0')\,\dd x_0'\,,
\end{equation}
where $\al$ is a parameter with real part greater that $n_0$, $h\in C^\infty_0(\Om)$, $x_0\in \Om$ and
\[
g_\al(s):=\frac{\pi^{1-\frac
{n_0}2}2^{1-\al}s^{\al-n}}{\Gamma(\frac\al2)\Gamma(\frac{\al-n_0}2+1)}
\]
It is obvious that the distribution-valued function $\al\mapsto R^\al_{\Om,x_0}$ is holomorphic in the half-plane $\Real\al>n$. A fundamental result in the theory of Riesz transforms~\cite{Ri51} states that this map can actually be holomorphically extended to the whole complex plane, the resulting distribution satisfying
\[
\square_0 R^\al_{\Om,x_0}=R^{\al-2}_{\Om,x_0}\,,\qquad R^0_{\Om,x_0}(h)=h(x_0)\,.
\]
for all $\al\in\CC$ and $x_0\in \Om$. Here $\square_0$ denotes the wave operator on $M_0$.

The expression of~\eqref{Riesz} in terms of scalar spherical means motivates the definition of a spinor analog of the Riesz transform in $\Om$, which is the bundle-valued distribution $\cR^\al_{\Om,x_0}\in \cD'(\Om,{\rm Hom}\,\Si(\Om))$ defined by setting
\begin{equation*}
  \cR^\al_{\Om,x_0}(\vp):=\int g_\al(r_0)^{\al-n_0}\,\cMo_{r_0}\vp(x_0)\,\dd\mu_0(r_0)
\end{equation*}
for any spinor field $\vp$ compactly supported in $\Om$. It is a straightforward consequence of the classical theory on scalar Riesz transforms that $ \cR^\al_{\Om,x_0}$ is an entire function of $\al$ and that
  \begin{equation*}
\cR^0_{\Om,x_0}(\vp)=\vp(x_0)\,.
\end{equation*}
In order to solve Eq.~\eqref{eqL} using this ``spinor Riesz transform'', we will need the following easy lemma, whose proof we sketch below.

\begin{lemma}\label{L.Riesz}
  For any multi-index $\Bla\in(\RR^+)^N$, the following statements hold:
  \begin{enumerate}
  \item If $\Real\al>n_0+2$, the function $f_\al(r_0,\Bla)$ defined by
    \begin{equation}\label{fal}
f_\al(r_0,\Bla):=\int \frac{w_0(r_0,\la_0)}{m^2+|\Bla|+\la_0}\,U_0g_\al(\la_0)\,\dd\vartheta_0(\la_0)
\end{equation}
satisfies the elliptic equation
\begin{equation}\label{eqfal}
  (m^2+|\Bla|-T_0)f_\al(\cdot,\Bla)=g_\al\,.
\end{equation}

  \item The spinor-valued distribution given when $\Real\al>n_0+2$ by
    \[
\tcR^\al_{\Om,x_0,\Bla}(\vp):=\int f_\al(r_0,\Bla)\,\cMo_{r_0}\vp(x_0)\,\dd\mu_0(r_0)
    \]
    defines an entire function of $\al$ and satisfies
    \[
\big(m^2+|\Bla|-D_0^2\big)\tcR^\al_{\Om,x_0,\Bla}=\cR^{\al-2}_{\Om,x_0}
    \]
for all $\al\in\CC$, $x_0\in\Om$ and $\Bla\in(\RR^+)^N$.
\end{enumerate}
\end{lemma}
\begin{proof}
  When $\al>n_0+2$, the function $g_\al$ is twice differentiable and in $L^1_{\rm loc}(\RR^+,\dd\mu_0)$, but not square-integrable, so we must begin by specifying in which sense $U_Mg_\al$ must be interpreted. Since the spectral transform $U_0:L^2(\RR^+,\dd\mu_0)\to L^2(\RR^+,\dd\vartheta_0)$ is unitary, it is not difficult to verify that it acts naturally on the space of tempered distributions as
  \[
(U_0u)(h):=u(U_0h)\,,\qquad \forall\,u\in \cD'(\RR^+)\,,\;\forall\, h\in C^\infty_0(\RR^+)\,.
  \]
  Therefore, Eq.~\eqref{fal} makes sense with the integral being understood in the sense of distributions. The fact that $f_\al(\cdot,\Bla)$ satisfies~\eqref{eqfal} then follows from Theorem~\ref{T.TM} and the density of $C^\infty_0(\RR^+)$ in $\cD'(\RR^+)$ through the  standard approximation argument.

  By Lemma~\eqref{L.DT2}, this ensures that
  \begin{align*}
    \big[\big(m^2+|\Bla|-D_0^2\big)\tcR^\al_{\Om,x_0,\Bla}\big](\vp)&=\tcR^\al_{\Om,x_0,\Bla}\big[\big(m^2+|\Bla|-D_0^2\big)\vp\big]\\
    &=\int f_\al(r_0,\Bla)\,\big[\cMo_{r_0}(m^2+|\Bla|-D_0^2)\vp\big](x_0)\,\dd\mu_0(r_0)\\
    &= \int (m^2+|\Bla|-T_0)f_\al(r_0,\Bla)\,\cMo_{r_0}\vp(x_0)\,\dd\mu_0(r_0)\\
    &=\int g_\al(r_0)\,\cMo_{r_0}\vp(x_0)\,\dd\mu_0(r_0)=\cR^{\al-2}_{\Om,x_0}(\vp)\,.
  \end{align*}
As  $m^2+|\Bla|-T_0$ has a compact inverse and $\cR^{\al-2}_{\Om,x_0}$ is an entire function of $\al$, it follows from the above relation that so is $\tcR^\al_{\Om,x_0,\Bla}$, thus completing the proof of the lemma.
\end{proof}

From our previous results, we can now easily derive the main result of this section, where we calculate the solution of Eq.~\eqref{eqL} in closed form.

\begin{theorem}\label{T.mainL}
  The spinor field on $\Om×\BM$ given by
  \[
  \psi(x_0,\Bx):=\int \Bw(0,\Bla)\,\Bw(\Br,\Bla)\, \tcR^2_{\Om,x_0,\Bla}\big(\BcM_\Br\phi(\cdot,\Bx)\big) \,\dd\Bmu(\Br)\,\dd\Btheta(\Bla)
  \]
  solves Eq.~\eqref{eqL}, for any compactly supported $\phi\in\bS(\Om×\BM)$. Moreover, the map $\phi\mapsto\psi$ defines the unique advanced fundamental solution of~\eqref{eqL}.
\end{theorem}
\begin{proof}
  It follows easily from Theorem~\ref{T.TM} and Lemmas~\ref{L.DT} and~\ref{L.Riesz} that $\psi$ is well defined and satisfies
  \begin{align*}
    \big(m^2- D_{M_0× \BM}^2\big)\psi(x_0,\Bx)&= \int \Bw(0,\Bla)\,\Bw(\Br,\Bla)\, (m^2-D_0^2)\tcR^2_{\Om,x_0,\Bla}\big(\BcM_\Br\phi(\cdot,\Bx)\big) \,\dd\Bmu(\Br)\,\dd\Btheta(\Bla)\\
    & -\int \Bw(0,\Bla)\,(T_1+\cdots+ T_N)\Bw(\Br,\Bla)\, \tcR^2_{\Om,x_0,\Bla}\big(\BcM_\Br\phi(\cdot,\Bx)\big) \,\dd\Bmu(\Br)\,\dd\Btheta(\Bla)\\
    & = \int \Bw(0,\Bla)\,\Bw(\Br,\Bla)\, (m^2+|\Bla|-D_0^2)\tcR^2_{\Om,x_0,\Bla}\big(\BcM_\Br\phi(\cdot,\Bx)\big) \,\dd\Bmu(\Br)\,\dd\Btheta(\Bla)\\
    &= \int \Bw(0,\Bla)\,\Bw(\Br,\Bla)\, \cR^0_{\Om,x_0}\big(\BcM_\Br\phi(\cdot,\Bx)\big) \,\dd\Bmu(\Br)\,\dd\Btheta(\Bla)\\
    & =\int \Bw(0,\Bla)\,U_{\BM}\big(\BcM\phi(x_0,\Bx)\big)(\Bla)\,\dd\Btheta(\Bla)=\phi(x_0,\Bx)\,,
  \end{align*}
  as required. The construction of $\cR^\al_{\Om,x_0,\Bla}$ ensures that $\phi\mapsto\phi$ defines an advanced fundamental solution of~\eqref{eqL}, the uniqueness of which is granted by the global hyperbolicity of $\Om$~\cite[Theorem 3.1.1]{BGP07}.
\end{proof}
\begin{remark}
  The retarded Green's function is also given by the above formula by substituting $J^+_\Om(x)$ for $J^-_\Om(x)$ in the definition of the Lorentzian spherical means.
\end{remark}

\section*{Acknowledgments}

A.E.\ is financially supported by a MICINN postdoctoral fellowship and
thanks McGill University for hospitality and support. A.E.'s research
is supported in part by the MICINN and the UCM--Banco Santander under
grants no.~FIS2008-00209 and~GR58/08-910556. The research of N.K. is
supported by NSERC grant RGPIN 105490-2004.


\end{document}